\documentclass[12pt]{article}

\usepackage{natbib}
\usepackage[nottoc,notlof,notlot]{tocbibind} 

\bibpunct{(}{)}{;}{a}{}{;}

\usepackage{xcolor}

\usepackage{breakcites}
\usepackage[]{hyperref}
\hypersetup{
    colorlinks=true,
    linkcolor=blue,
    citecolor=blue,
}

\usepackage{graphicx}
\graphicspath{ {images/} }

\usepackage{amsmath}
\usepackage{bbm}
\usepackage{amsfonts}%
\usepackage{amssymb}%
\usepackage{amsthm}
\usepackage{geometry} % to change the page dimensions
\usepackage{graphicx} % support the \includegraphics command and options
\usepackage{framed}
\usepackage{caption}
\usepackage{subcaption}

\usepackage{booktabs} % for much better looking tables
\usepackage{bigstrut} % for Excel2Latex
\usepackage{array} % for better arrays (eg matrices) in maths
\usepackage{paralist} % very flexible & customisable lists (eg. enumerate/itemize, etc.)
\usepackage{verbatim} % adds environment for commenting out blocks of text & for better verbatim
\usepackage{setspace}

\usepackage[singlelinecheck=false]{caption}
\usepackage{tabularx}

\usepackage{multirow}

\title{\setstretch{1} Some Finite Sample Properties of the Sign Test\footnote{I am grateful to Ivan Canay, Joel Horowitz and Azeem Shaikh for guidance and helpful comments.}}

\author{Yong Cai\footnote{Becker Friedman Institute, University of Chicago. E-mail: yongcai@uchicago.edu}
}
 \date{\parbox{\linewidth}{\centering%
  \today\endgraf}} % Activate to display a given date or no date (if empty),
\theoremstyle{definition}

\newtheorem{theorem}{Theorem}
\newtheorem{proposition}{Proposition}
\newtheorem{corollary}{Corollary}
\newtheorem{assumption}{Assumption}

\newtheorem{remark}{Remark}

\begin{document}
\maketitle
%\centerline{\large PRELIMINARY AND INCOMPLETE}
\begin{abstract} \setstretch{1}\noindent
This paper contains two finite-sample results concerning the sign test. First, we show that the sign-test is unbiased with independent, non-identically distributed data for both one-sided and two-sided hypotheses. The proof for the two-sided case is based on a novel argument that relates the derivatives of the power function to a regular bipartite graph. Unbiasedness then follows from the existence of perfect matchings on such graphs. Second, we provide a simple theoretical counterexample to show that the sign test over-rejects when the data exhibits correlation. Our results are useful for understanding the properties of approximate randomization tests in settings with few clusters. 

\end{abstract}

%{\color{red} Can also permute location and compare signs}. 

\section{Introduction}

In recent years, randomization tests have become increasingly popular as a method for inference, due in no small part to the fact that they control size exactly even in finite sample. For example, \cite{ai2017} and \cite{young2019} advocate their in policy evaluation for this reason. Exactness has also been useful in settings with a small number of ``effective observations". Here, 
tests are set up to be asymptotically equivalent to randomization tests with a finite number of observations that are independently but not identically distributed. Exactness of randomization tests then yields size control. Such as approach has been fruitful in for inference for cluster-dependent data with few clusters \citep{crs2017, css2019} and specification tests for RDD involving few units near the cut-off (\citealt{canay2018approximate, bugni2021testing}), among others. 

In the above examples, randomization tests with a fixed number of observations are often used directly or indirectly for inference. It is therefore useful to understand their finite sample properties. \cite{hoeffding1952} showed that randomization tests control size in this setting, provided that the randomization hypothesis holds -- that is, that the data satisfies certain symmetries. The power of randomization tests is less well documented, particularly when observations are not identically distributed. It is challenging to analyze the power of randomization tests in finite sample due to their combinatorial nature. However, allowing for non-identical observations is important. In \cite{crs2017}, each observation corresponds to a cluster, so accommodating non-identical observations is necessary to allow for non-identical clusters, a keystone of the clustering literature. Additionally, it is useful to assess the sensitivity of randomization tests to violations of the randomization hypothesis. In the clustering context, researchers rarely know the cluster structure ex ante. Misspecify the cluster structure may lead to (approximate) randomization tests with dependence between observations. Understanding how randomization tests behave under violations of independence is thus of relevance to applied researchers. 

In this paper, we attempt to make headway on the above issues by focusing on the sign test, which is a test for the median of a collection of random variables. If these random variables are symmetrically distributed, it is also a test for the mean. In addition to being an especially tractable randomization test, the sign test is also relevant given its use in financial econometrics to test for abnormal returns in event studies (\citealt{corrado1992specification}), in RDD to test the continuity of the running variable (\citealt{bugni2021testing}), or to test for the level of clustering in computing clustered standard error (\citealt{cai2023modified}), among other applications.

We document two properties of the sign test in finite samples. Our first result shows that when observations are independent but not identically distributed, the sign test is unbiased against one and two-sided alternatives. Our proof for the two sided case is based on a novel argument that relates the derivative of the power function to a regular bipartite graph. Hall's Marriage Theorem, which implies the existence of a perfect matching for regular bipartite graphs, then yields unbiasedness.  Our second result shows that sign tests with correlated normal random variables strictly over-rejects. This provides a clear counterexample, showing that it is important to account for correlation within units when conducting sign tests and randomization tests more generally. 

This paper complements the literature on the properties of the sign test. It is well-known since \cite{hoeffding1952} that the sign test for the median has exact size in finite sample, even when observations non-identically distributed. On the issue of finite sample power, under the assumption that observations are identically distributed, the sign test is known to be the uniformly most powerful unbiased test for one- and two-sided hypotheses about the median (see Section 4.9 and Theorem 4.4.1 in \citealt{lr2005}). \cite{dixon1953power} computes the power of the sign test under alternatives which are identically, normally distributed. However, to our knowledge, there does not exist power results once we allow for observations that are not identically distributed. In particular, it is not clear how the standard approach generalizes to this case (see Remark \ref{remark--standard_approach}). Our paper provides a novel argument to show unbiasedness. On the issue of dependence, \cite{gastwirth1971effect} shows that the sign test over-rejects asymptotically when observations exhibit first order auto-correlation. Our counterexample, based on equicorrelated normal random variables, demonstrates the same phenomenon simply in the finite sample setting.  

The rest of this paper is organized as follows. Section 2 presents notation and describes the sign test. Section 3 contains our main results. Section 4 concludes. All proofs are contained in the appendix. 

\section{Set up and Notation}\label{section--setup}

Let $\mu = \text{Med}(X_1) = ... = \text{Med}(X_q)$. 

\subsubsection*{Two-Sided Test}

We want to test $$H_0: \mu = \mu_0 \quad \text{against} \quad H_A: \mu \neq \mu_0~.$$ First define $Y = \left(X_1 - \mu_0, ..., X_q - \mu_0\right)'$. Next, define the two-sided test statistic
\begin{equation}\label{equation--twosidedteststat}
	T(Y) = \left\lvert \sum_{i=1}^q \text{sgn}\left\{ Y_i\right\} \right\rvert~,
\end{equation}
Next, denote by $\mathbf{G}$ the set of sign changes. $\mathbf{G}$ can be identified with the set of $g \in \{-1, 1\}^{q}$ so that $gY = \left(g_1Y_1, ... ,g_qY_q\right)'$. The two-sided sign test $\phi_{2,n}$ is the randomization test that rejects the null hypothesis when $T(Y)$ takes on extreme value relative to $T(gY)$. It proceeds as follows. 

Define $M = |\mathbf{G}|$ and let:
\begin{equation*}
T^{(1)}(Y) \leq T^{(2)}(Y) \leq ... \leq T^{M}(Y)
\end{equation*}
be the ordered values of $T(gY)$ as $g$ varies in $\mathbf{G}$. For a fixed nominal level $\alpha$, let $k$ be defined as
\begin{equation*}
k = \lceil (1-\alpha) M \rceil~,
\end{equation*}
where $\lceil x \rceil$ denotes the smallest integer greater or equal to $x$. In addition, define:
\begin{align*}
M^+(Y) & = \sum_{j=1}^{M} I\{T^{(j)}(Y) > T^{(k)}(Y)\} \\
M^0(Y) & = \sum_{j=1}^{M} I\{T^{(j)}(Y) = T^{(k)}(Y)  \}
\end{align*}
and set
\begin{equation*}
a(Y) = \frac{M\alpha - M^+(Y)}{M^0(Y)}~.
\end{equation*}
We can then define the two-sided sign test as:
\begin{align} \label{test_standard}
\phi_{2,n}(Y) = \begin{cases}
1 & \mbox{ if }T(Y) > T^{(k)}(Y), \\
a(Y) & \mbox{ if }T(Y) = T^{(k)}(Y), \\
0 & \mbox{ otherwise.}
\end{cases}
\end{align}
In words, this test rejects the null hypothesis with certainty when $T(Y) > T^{(k)}(Y)$. When $T(Y) = T^{(k)}(Y)$, it rejects the null hypothesis with probability $a(Y)$. The test does not reject when $T(Y) < T^{(k)}(Y)$. Suppose $\{Y_i\}_{i=1}^q$ do not have point mass at $\mu$. Then under the null hypothesis, $\text{sgn}\{Y_i\}$ is a Rademacher random variable and
$$\text{Unif}\left(T^{(1)}(Y), T^{(2)}(Y), ..., T^{(M)}(Y)\right) \, \bigg\lvert \, Y \sim \text{Binomial}\left(q, \frac{1}{2}\right)~.$$
As such, the critical value for our test can equivalently be written in terms of the quantiles of the $\text{Binomial}(q, \frac{1}{2})$ distribution. 

\subsubsection*{One-Sided Test}

Suppose we are instead interested in the hypothesis:
$$H_0: \mu = \mu_0 \quad \text{against} \quad H_A: \mu \geq \mu_0~.$$ 
We perform the one-sided sign test with the test statistic: 
\begin{equation}
	S(Y) = \sum_{i=1}^q \text{sgn}\left\{ Y_i\right\} ~.
\end{equation}
The rest of the testing procedure is identical to that of the two-sided test, except with $S(Y)$ taking the place of $T(Y)$. Let the corresponding test by $\phi_{1,n}$. 

In the following sections, $E_\mu[\,\cdot\,]$ denotes expectation evaluated with respect to the data-generating process with median $\mu$. 

\section{Results}

\subsection{Unbiasedness of the Sign Test}
In this subsection, we show that the sign test is unbiased against two-sided alternatives in finite sample. Our result implies that the approximate sign test is asymptotically unbiased. 
\begin{assumption} \label{assumption--independence}
Let $\{Y_i\}_{i=1}^q$ be independent, not necessarily identical random variables with median $\mu$.
\end{assumption}
It is well known that the sign test (and randomization tests more generally) has exact $p$-values in finite sample: 
\begin{theorem}[\citealt{hoeffding1952}]\label{theorem--exactness}
	Given assumption \ref{assumption--independence}, suppose $\mu = \mu_0$. Then $E_{\mu_0}[\phi_n(Y)] = \alpha$. 
\end{theorem}

It is also known that the sign test against one-sided alternatives is unbiased even with random variables that are not identically distributed (see Lemma 5.9.1 in \cite{lr2005} and also Appendix \ref{appendix--onesided}). Compared to one-sided alternatives, tests against two-sided alternatives are arguably more relevant, given the state of applied econometrics today. Our unbiasedness result for two-sided alternatives states that:

\begin{proposition}\label{theorem--unbiased}
	Given Assumption \ref{assumption--independence}, suppose $\mu \neq \mu_0$. Then,
	\begin{itemize}
		\item[(a)] $E_{\mu}[\phi_{2,n}(Y)] \geq \alpha$.
		\item[(b)] If $\alpha < 1$ and at least one of $\{Y_i\}_{i=1}^n$ is continuously distributed around $\mu$, the inequality in (a) is strict. 
		\item[(c)] Furthermore, as $|\mu - \mu_0| \to \infty$, $E_\mu[\phi_{2,n}(Y)] \to 1$. 
	\end{itemize}
	The same results hold for $\phi_{1,n}$ if $\mu > \mu_0$.
\end{proposition}
To show that $E_{\mu}[\phi_{t,n}(Y)] > \alpha$, we need to avoid the case in which $P_\mu(Y_i \leq \mu_0) = P_\mu(Y_i > \mu_0)$ for all $i$. This might occur if for instance all the random variables have zero density in a neighbourhood of $\mu$. This definitely cannot happen for continuous random variable, as is the case with approximate sign tests. Note that power cannot be 1 for any finite parameter value. This is because we are conducting inference using finite number of ``effective observations".

Unbiasedness in the one-sided test follows straightforwardly from monotonicity of the test statistic in $\mu_0$. Our argument is similar in spirit to Lemma 5.9.1 in \cite{lr2005}, which concerns the permutation test for two-sample comparison of means. Our argument for the two-sided test is based takes a novel approach that may be of interest. It relates derivatives of the power function to a regular bipartite graph. On this graph, each node represents a term in the derivative and an edge represents an inequality relationship. A perfect matching, which exists by Hall's Marriage Theorem, then allows us to sign the derivative. 

\begin{remark}\label{remark--standard_approach}
	When observations are identically distributed, the approach of \cite{lr2005} to prove unbiasedness proceeds by showing that the distribution of signs under is a member of the one-parameter exponential family. The sign test can then be shown to be the uniformly most powerful test, which in turn implies unbiasedness. This argument is presented in Section 4.9, drawing on Theorems 4.4.1 and 3.7.1. When observations are allowed to be non-identically distributed, the distribution of signs under the null hypothesis is not a member of the one-parameter exponential family. In particular, under a given alternative hypothesis, each random variable $\text{sgn}\{Y_i\}$ will be $1$ with a different probability $p_i$. It is therefore not clear how to adapt this approach to the non-identical distributions. 
\end{remark}

%This is because comparing the distributions with true and misspecified medians in terms of first order stochastic dominance -- the approach of \cite{lr2005} -- is no longer sufficient. Instead, we apply results in graph matching to provide lower bounds for power. 

The corollary below states that an approximate randomization test is unbiased as well. It follows immediately by the continuous mapping theorem.  
\begin{corollary}
	Suppose we have variables $S_n \to \text{N}(\mu\cdot \iota, D)$, where $D$ is a $q \times q$ diagonal matrix and $\iota$ is the $q$-vector of 1's. 
	Then for $t \in \{1,2\}$, 
	\begin{itemize}
		\item[(a)] If $\alpha < 1$, $\lim_{n \to \infty }E[\phi_{t,n}(S_n)] > \alpha$
		\item[(b)] As $|\mu - \mu_0| \to \infty$, $\lim_{n \to \infty} E_\mu[\phi_{t,n}(S_n)] \to 1$
	\end{itemize}
\end{corollary}

\subsection{Size Distortion under Equicorrelation}

Here, we study size distortion of the sign test with normal, equicorrelated random variables that are wrongly assumed to be independent. This is a counterexample showing that over-rejection could be a consequence when practitioners using approximate sign test gets the cluster structure wrong.

\begin{assumption} \label{assumption--equicorrelation}
Let $Y \sim \text{N}(\mu, (1-\rho)\mathbf{I}_q + \rho \cdot \iota \iota^T)$ be equicorrelated normal random variables. 
\end{assumption}

Consider again test \ref{test_standard}. We show that under $H_0$, when $\mu = \mu_0$, the test will over-reject:
\begin{proposition}\label{theorem--equicorrelation}
	Given Assumption \ref{assumption--equicorrelation},  suppose that $\mu = \mu_0$, $\rho > 0$, $q > 1$ and $\alpha < 0.25$. Then $$E[\phi_{2,n}(Y)] > \alpha~.$$ 	Furthermore, as $\rho \to 1$, $E[\phi_n(Y)] \to 1$. 
\end{proposition}
In other words, the sign test could mistake correlation for violation of the null. Intuitively, the sign test checks for balance of the signs around the mean. If the variables are positively correlated, there will be an unusually large number of random variables with the same sign relative to independence, so that the test wrongly rejects. Our result is stated for the two-sided test, though the analogous result for the one-sided test follows straightforwardly. Our counterexample complements a large literature showing that the usual $t$- and Wald tests can experience significant size distortion when dependence is ignored (see for instance \citealt{bdm2004, cgm2008}).

To further understand size distortion under equicorrelation, we compute asymptotic power numerically. This is possible because the size (in particular equation \ref{equation--gb2009} in appendix A.2) can be expressed as a scalar integral that can be quickly and accurately computed using Gauss-Hermite quadratures. We evaluate the power of the test at $5\%$ and $10\%$ levels of significance using 1000 nodes. The results are presented in figures \ref{fig:powerrand5random} and \ref{fig:powerrand10random} respectively.

For the test at 5\%-level, we see that for the $\rho$'s considered, size starts at $5\%$ when the number of sub-clusters is 1. It then monotonically increases as we increase the number of sub-clusters. Furthermore, fixing the number of sub-clusters, the size of the test increases also increases monotonically in the correlation across the sub-clusters. The test at 10\%-level yields similar results. All in all, the numerical results concord with that of Proposition \ref{theorem--equicorrelation}. 

\begin{figure}[]
\centering
\includegraphics[width=1\linewidth]{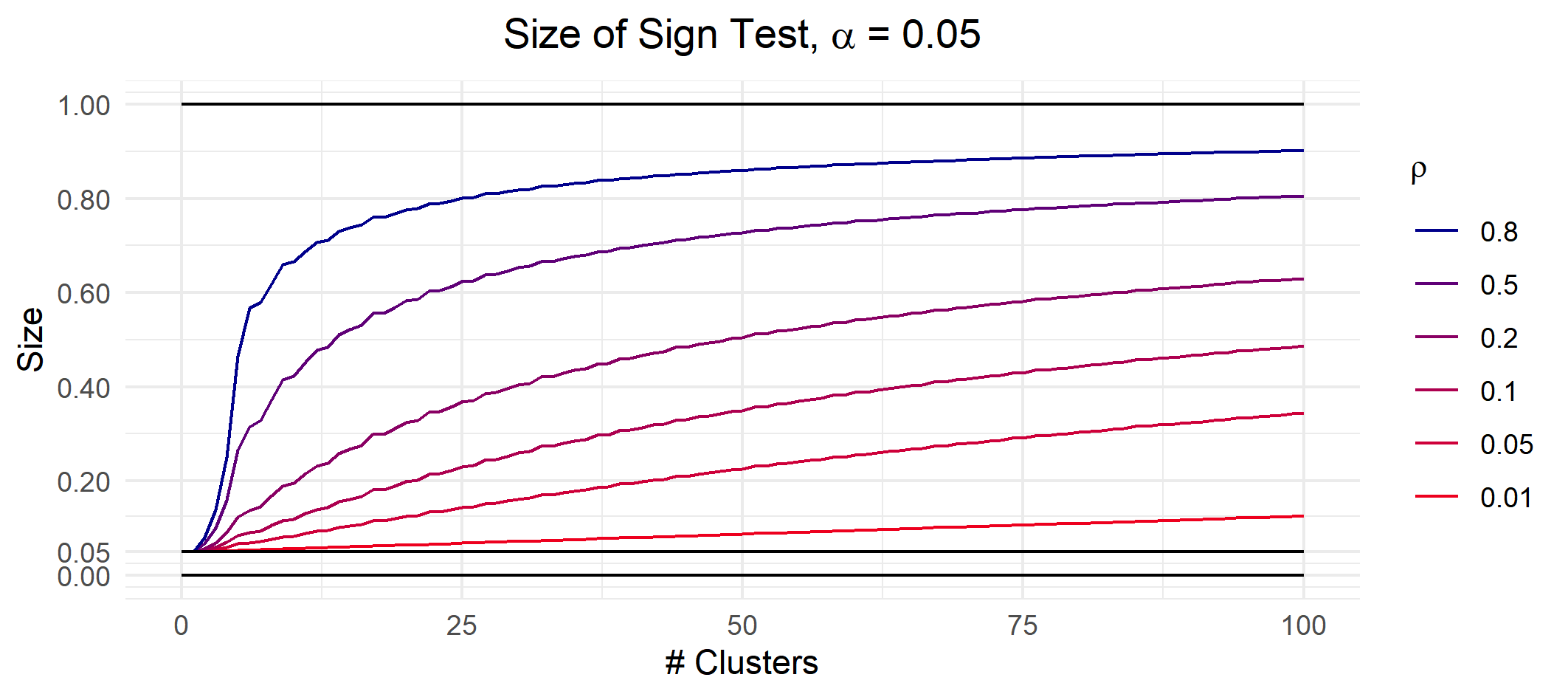}
\caption{Size of the 5\%-level permutation test with the equicorrelated model under various values of $\rho$. Over-rejection rate test increases monotonically in the number of sub-clusters and in $\rho$, the amount of correlation within each cluster. Power is evaluated using Gauss-Hermite quadratures with 1000 nodes.}
\label{fig:powerrand5random}
\end{figure}

\begin{figure}[]
\centering
\includegraphics[width=1\linewidth]{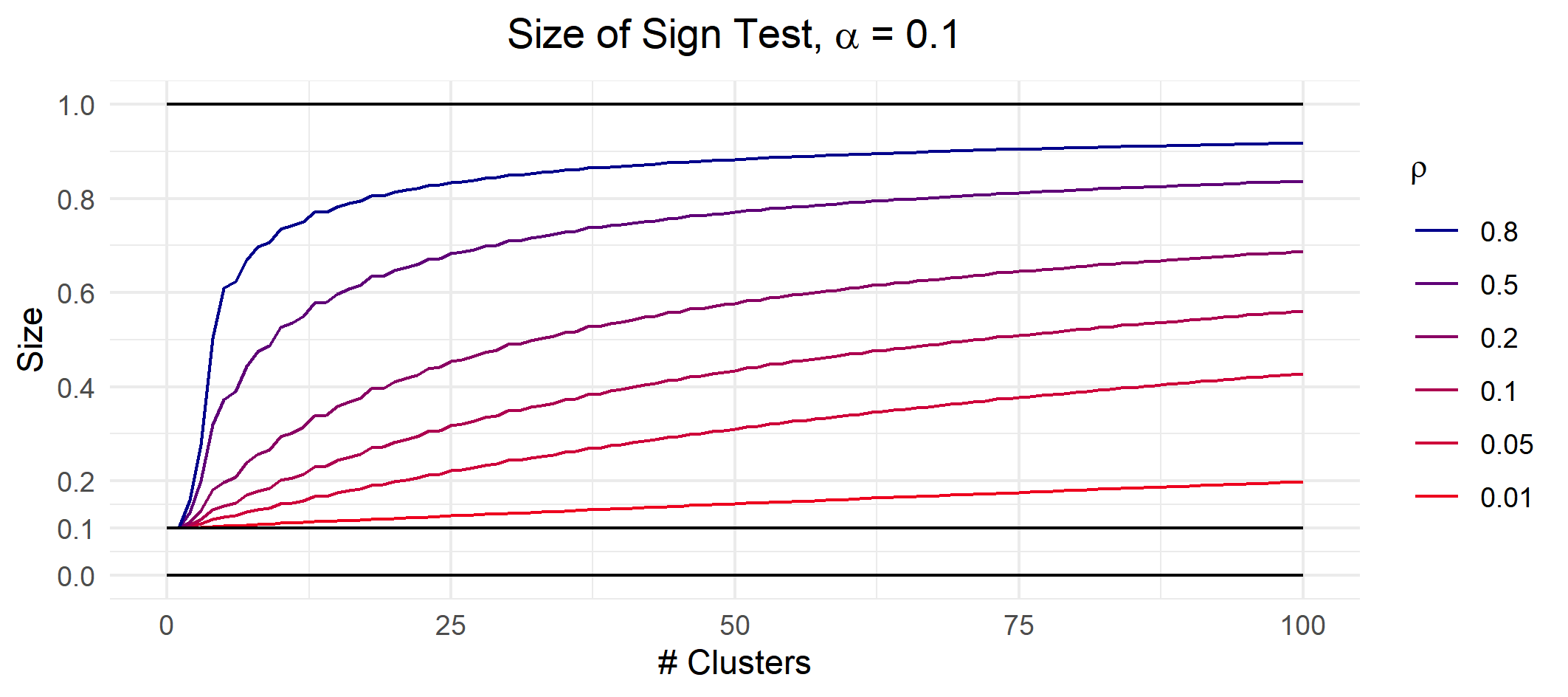}
\caption{Size of the 10\%-level permutation test with the equicorrelated model under various values of $\rho$. Over-rejection rate test increases monotonically in the number of sub-clusters and in $\rho$, the amount of correlation within each cluster. Power is evaluated using Gauss-Hermite quadratures with 1000 nodes.}
\label{fig:powerrand10random}
\end{figure}

\subsubsection{Size Distortion under Minimal Correlation}

Equicorrelation appears to be a strong assumption, since all of $Y$ is assume to be correlated. Suppose only the first two entries of $Y$ are correlated. In turns out that size distortion arises as well. 

\begin{assumption} \label{assumption--minimalcorrelation}
Let $Y \sim \text{N}(\mu, \Sigma)$, where 
\begin{equation}
	\Sigma = \begin{pmatrix}
	1	& \rho & 0 & \cdots & 0 \\
	\rho& 1 &  0 & \cdots & 0 \\
	0   & 0 & 1 & \cdots & 0 \\		
	\vdots& \vdots & \vdots & \ddots & \vdots \\
	0   & 0    & 0 & \cdots & 1 \\
	\end{pmatrix}~.
\end{equation}
\end{assumption}

Our earlier result also applies to this case: 
\begin{proposition}\label{theorem--minimalcorrelation}
	Given Assumption \ref{assumption--minimalcorrelation},  suppose that $\mu = \mu_0$, $\rho > 0$, $q > 1$ and $\alpha < 0.25$. Then $$E[\phi_{2,n}(Y)] > \alpha~.$$
\end{proposition}
  
\section{Conclusion}
We document two finite sample properties of the sign test. We show that it is unbiased against one- and two-sided alternatives, even with non-identically distributed observations, but can experience size distortion if the samples are positively correlated. In the context of approximate sign tests under fixed cluster asymptotics, the latter result shows that misspecifying the cluster structure will lead to over-rejection of the null. Our result can be a useful first step towards a more general theory for the finite sample properties of randomization tests. 

\hfill \\

\bibliographystyle{chicago}
\bibliography{clusters}

\begin{thebibliography}{}

\bibitem[\protect\citeauthoryear{Arbuthnot}{Arbuthnot}{1710}]{arbuthnot1710ii}
Arbuthnot, J. (1710).
\newblock Ii. an argument for divine providence, taken from the constant
  regularity observ'd in the births of both sexes. by dr. john arbuthnott,
  physitian in ordinary to her majesty, and fellow of the college of physitians
  and the royal society.
\newblock {\em Philosophical Transactions of the Royal Society of
  London\/}~{\em 27\/}(328), 186--190.

\bibitem[\protect\citeauthoryear{Athey and Imbens}{Athey and
  Imbens}{2017}]{ai2017}
Athey, S. and G.~W. Imbens (2017).
\newblock {The State of Applied Econometrics: Causality and Policy Evaluation}.
\newblock {\em Journal of Economic Perspectives\/}~{\em 31(2)}, 3--32.

\bibitem[\protect\citeauthoryear{Bertrand, Duflo, and Sendhil}{Bertrand
  et~al.}{2004}]{bdm2004}
Bertrand, M., E.~Duflo, and M.~Sendhil (2004).
\newblock {How Much Should We Trust Differences-in-Differences Estimates?}
\newblock {\em Quarterly Journal of Economics,\/}~{\em 119}, 249--275.

\bibitem[\protect\citeauthoryear{Bugni and Canay}{Bugni and
  Canay}{2021}]{bugni2021testing}
Bugni, F.~A. and I.~A. Canay (2021).
\newblock Testing continuity of a density via g-order statistics in the
  regression discontinuity design.
\newblock {\em Journal of Econometrics\/}~{\em 221\/}(1), 138--159.

\bibitem[\protect\citeauthoryear{Cai}{Cai}{2023}]{cai2023modified}
Cai, Y. (2023).
\newblock A modified randomization test for the level of clustering.
\newblock {\em Journal of Business \& Economic Statistics\/}, 1--13.

\bibitem[\protect\citeauthoryear{Cameron, Gelbach, and Miller}{Cameron
  et~al.}{2008}]{cgm2008}
Cameron, A.~C., J.~B. Gelbach, and D.~L. Miller (2008).
\newblock {Bootstrap-Based Improvements for Inference with Clustered Errors}.
\newblock {\em The Review of Economics and Statistics,\/}~{\em 90\/}(3),
  414--427.

\bibitem[\protect\citeauthoryear{Canay and Kamat}{Canay and
  Kamat}{2018}]{canay2018approximate}
Canay, I.~A. and V.~Kamat (2018).
\newblock Approximate permutation tests and induced order statistics in the
  regression discontinuity design.
\newblock {\em The Review of Economic Studies\/}~{\em 85\/}(3), 1577--1608.

\bibitem[\protect\citeauthoryear{Canay, Romano, and Shaikh}{Canay
  et~al.}{2017}]{crs2017}
Canay, I.~A., J.~P. Romano, and A.~M. Shaikh (2017).
\newblock {Randomization Tests under an Approximate Symmetry Assumption}.
\newblock {\em Econometrica,\/}~{\em 85\/}(3), 1013--1030.

\bibitem[\protect\citeauthoryear{Canay, Santos, and Shaikh}{Canay
  et~al.}{2019}]{css2019}
Canay, I.~A., A.~Santos, and A.~M. Shaikh (2019).
\newblock {The wild bootstrap with a ''small" number of ''large" clusters}.
\newblock {\em Review of Economics and Statistics (forthcoming)\/}.

\bibitem[\protect\citeauthoryear{Corrado and Zivney}{Corrado and
  Zivney}{1992}]{corrado1992specification}
Corrado, C.~J. and T.~L. Zivney (1992).
\newblock The specification and power of the sign test in event study
  hypothesis tests using daily stock returns.
\newblock {\em Journal of Financial and Quantitative analysis\/}~{\em 27\/}(3),
  465--478.

\bibitem[\protect\citeauthoryear{Dixon}{Dixon}{1953}]{dixon1953power}
Dixon, W.~J. (1953).
\newblock Power functions of the sign test and power efficiency for normal
  alternatives.
\newblock {\em The Annals of Mathematical Statistics\/}, 467--473.

\bibitem[\protect\citeauthoryear{Gastwirth and Rubin}{Gastwirth and
  Rubin}{1971}]{gastwirth1971effect}
Gastwirth, J.~L. and H.~Rubin (1971).
\newblock Effect of dependence on the level of some one-sample tests.
\newblock {\em Journal of the American Statistical Association\/}, 816--820.

\bibitem[\protect\citeauthoryear{Genz and Bretz}{Genz and Bretz}{2009}]{gb2009}
Genz, A. and F.~Bretz (2009).
\newblock {\em Computation of Multivariate Normal and t Probabilities}.
\newblock Springer-Verlag Berlin Heidelberg.

\bibitem[\protect\citeauthoryear{Hoeffding}{Hoeffding}{1952}]{hoeffding1952}
Hoeffding, W. (1952).
\newblock {The Large-Sample Power of Tests Based on Permutations of
  Observations}.
\newblock {\em Annals of Mathematical Statistics\/}~{\em 23(2)}, 169--192.

\bibitem[\protect\citeauthoryear{Lehmann and Romano}{Lehmann and
  Romano}{2005}]{lr2005}
Lehmann, E. and J.~P. Romano (2005).
\newblock {\em Testing Statistical Hypotheses}.
\newblock Springer.

\bibitem[\protect\citeauthoryear{Young}{Young}{2019}]{young2019}
Young, A. (2019).
\newblock {Channelling Fisher: Randomization Tests and the Statistical
  Insignificance of Seemingly Significant Experimental Results}.
\newblock {\em Quarterly Journal of Economics\/}~{\em 134(2)}, 557--598.

\end{thebibliography}

\clearpage

\appendix

\section*{Appendices}\label{appendix--proofs}

\section{Proofs}

We first state a simple proposition:
\begin{proposition} \label{proposition--generalpower}
	Suppose $2^{-(q-1)}\sum_{r=0}^{m-1} {q \choose r}  <  \alpha \leq 2^{-(q-1)}\sum_{r=0}^{m} {q \choose r} \leq \frac{1}{2}$, where $\sum_{r=0}^{m-1} {q \choose r}$ is understood to be $0$ if $m-1 < 0$. Then,
	\begin{equation*}
				\mathbb{E}[\phi_n] = \sum_{r = 0}^{m-1} P_Y(r) + \gamma P_Y(m)
	\end{equation*}
	where $\gamma = \frac{ \alpha2^{q-1} - \sum_{r \leq m-1} {q \choose r} }{{q \choose m}}$ and $P_Y(r)$ is the probability that $Y-\mu\cdot \iota$ has exactly $r$ positive or exactly $r$ negative realisations. 
\end{proposition}

\begin{proof}
	Our randomization test rejects the null hypothesis with certainty if  ${T}(Y)  > {T}^{(k)}(Y)$. If ${q \choose m-1} 2^{-(q-1)} <  \alpha \leq {q \choose m} 2^{-(q-1)}$, $T^{(k)}(S) = q-2m$. Accordingly, $Y$'s with fewer than $m$ positive entries or fewer than $m$ negative entries are considered extreme relative to the randomization distribution. We reject with certainty whenever ${T}(Y) > (q-2m)$. When ${T}(Y)  = {T}^{(k)}(Y)$, we reject the null hypothesis with probability $a(Y)$, which specializes in this context to $\gamma = \frac{ \alpha2^q - \sum_{r \leq m-1} {q \choose r} }{{q \choose m}}$. Summing these probabilities we have the above formula.  
\end{proof}

\noindent Note that since we count both positive and negative realisations, WLOG, $m \leq q/2$. 

\subsection{Proof for Proposition \ref{theorem--unbiased}}

\subsubsection*{Two-Sided Case}

WLOG, we prove the result assuming that $\mu > \mu_0$, since the problem is symmetric around $\mu$. Then $p_i \geq \frac{1}{2}$ for all $i$. Suppose also for now that $q \geq 2m + 1$. 

Let $C(A, r)$ be all the possible ways of choosing $r$ elements of the set $A$ and let $$p_i = P(Y_i > -(\mu - \mu_0))~.$$ Then we can write:
\begin{align*}
	P_Y(r) = \sum_{c \in C([q], q-r)} \prod_{i \in c} p_i \prod_{i \notin c} (1-p_i) + \sum_{c \in C([q], r)} \prod_{i \in c} p_i \prod_{i \notin c} (1-p_i)
\end{align*}
where it is understood that $C([q], 0)$ is the empty set. This follows from simple counting. Next, consider their first partial derivative with respect to $p_j$. For $r \geq 1$, this is:
\begin{align*}
	\frac{\partial }{\partial p_j} P_Y(r) = & \sum_{c \in C([q] \setminus \{j\}, q-r-1)} \prod_{i \in c} p_i \prod_{i \notin c} (1-p_i) \quad - \sum_{c \in C([q] \setminus \{j\}, q-r)} \prod_{i \in c} p_i \prod_{i \notin c} (1-p_i) \\
	&  + \sum_{c \in C([q] \setminus \{j\}, r-1)} \prod_{i \in c} p_i \prod_{i \notin c} (1-p_i) \quad - \sum_{c \in C([q] \setminus \{j\}, r)} \prod_{i \in c} p_i \prod_{i \notin c} (1-p_i)
\end{align*}
For $r = 0$, it is simply:
\begin{align*}
	\frac{\partial }{\partial p_j} P_Y(0) & =  \sum_{c \in C([q] \setminus \{j\}, q-1)} \prod_{i \in c} p_i \prod_{i \notin c} (1-p_i)
	\quad - \sum_{c \in C([q] \setminus \{j\}, r)} \prod_{i \in c} p_i \prod_{i \notin c} (1-p_i)  \\
    & = \prod_{i \neq j} p_i - \prod_{i \neq j} (1-p_i)
\end{align*}
Suppose for now that $\gamma = 1$. Then, given the telescoping nature of the above terms, we can write:
\begin{align*}
	\frac{\partial }{\partial p_j} \sum_{r=0}^m P_Y(r)  = \sum_{c \in C([q] \setminus \{j\}, q-m-1)} \prod_{i \in c} p_i \prod_{i \notin c} (1-p_i) - \sum_{c \in C([q] \setminus \{j\}, m)} \prod_{i \in c} p_i \prod_{i \notin c} (1-p_i)
\end{align*}

We want to prove that the derivative is positive. Suppose for each $c \in C\left([q] \setminus \{j\} \,, \, q-m-1\right)$, we can find a $c' \in C\left([q] \setminus \{j\} \,, \, m\right)$ such that $i \notin c \Rightarrow i \not\in c'$. Then, using the fact that $p_i \geq 1- p_i$, we must have that
\begin{align*}
  \prod_{i \in c} p_i \prod_{i \notin c} (1-p_i) - \prod_{i \in c'} p_i \prod_{i \notin c'} (1-p_i) = & \prod_{i \notin c} (1-p_i) \left( \prod_{i \in c} p_i -  \prod_{i \in c' \cap c} p_i \prod_{i \in c, \, i \notin c'} (1-p_i)   \right) \\
 = & \begin{cases}
 > 0 & \mbox{ if at least one $p_i > 1/2$} \\
 = 0 & \mbox{ if } p_i = 0 \, \, \forall i
 \end{cases}~.
\end{align*}

Note that because $q \geq 2m +1$, for each $c$, there are ${q-m-1 \choose m } \geq 1$ potential $c'$. Similarly, for each $c'$, there are ${q-m-1 \choose m}$ potential $c$. Consider the (undirected) bipartite graph in which each $c \in C\left([q] \setminus \{j\} \,, \, q-m-1\right)$ is connected to eligible $c' \in C\left([q] \setminus \{j\} \,, \, m\right)$. Our problem is equivalent to finding a perfect matching in this ${q-m-1 \choose m}$-regular bipartite graph.

By Hall's Marriage Theorem, it suffices to check that for all $W \subseteq C\left([q] \setminus \{j\} \,, \, q-m-1\right)$, $|N(W)| \geq |W|$, where $N(W)$ is the set of 1-neighbours to $W$. Let $E(W)$ be the set of edges from $W$ to $N(W)$ and let $E(N(W))$ be the set of edges from $N(W)$ to $C\left([q] \setminus \{j\} \,, \, q-m-1\right)$. Next, note that $E(W) \subset E(N(W))$ by definition of $N(W)$. As such, we write:
\begin{align*}
	|E(W)| \leq E(N(W)) \Leftrightarrow {q-m-1 \choose m}|W| \leq {q-m-1 \choose m}|N(W)| \Leftrightarrow |W| \leq |N(W)|
\end{align*}
where the first equivalence follows from regularity of the graph. Hence, by \ref{proposition--generalpower}
\begin{align*}
\frac{\partial}{\partial \mu} E_\mu[\phi_n] & = \sum_{r = 0}^{m} \frac{\partial}{\partial \mu}P_Y(r)  =     \sum_{i=1}^q \frac{\partial p_i}{\partial \mu} \cdot \left(\sum_{r = 0}^{m} \frac{\partial}{\partial p_i}   P_Y(r)  \right)   \geq 0~.
\end{align*}
Here, the inequality is strict as long as there is at least one $i$ such that $p_i > \frac{1}{2}$ and $\frac{\partial p_i}{\partial \mu} > 0$. The strict inequality obtains if at least one random variable is continuous with unbounded support. 
Finally, we have that
\begin{align*}
	E_\mu[\phi_n] = E_{\mu_0}[\phi_n] + \int_{\mu_0}^{\mu}  \frac{\partial}{\partial \mu} E_u[\phi_n] \,    du \geq \alpha
\end{align*}
since the first term is $\alpha$ by exactness of the sign test and the second term is the integral of a non-negative function. 

Finally, note that as $\mu - \mu_0 \to \infty$, $p_i \to 1$. Hence $P_Y(r) \to 0$ for all $r \neq 0$. $P_Y(0) \to 1$.

Now suppose $\gamma > 0$. From \ref{proposition--generalpower}, we see that the power of test with $\gamma > 0$ is a linear combination of two tests each with power at least $\alpha$. Hence, we have that this test is unbiased as well. Furthermore, if either of the two test has power strictly greater than $\alpha$, so does the test with $\gamma > 0$. 

We have shown our result for $q \geq 2m + 1$. Note that $q \geq 2m$ by definition, since this is a 2-sided test. Suppose $q = 2m$. Then this is the trivial test that rejects 100\% of the time for a significance level of $0$. Our test can still have strictly increasing power, but this comes from the $(m-1)$ term. As such, the proposition continues to hold as long as $\alpha < 1$. 

\subsubsection*{One-Sided Case}

Suppose that $\mu > \mu_0$. Then the test statistic takes value:
\begin{equation*}
	S(Y) = \sum_{i=1}^q \text{sgn}\left\{ X_i -\mu_0\right\} \geq \sum_{i=1}^q \text{sgn}\left\{ X_i -\mu \right\} =: \tilde{S}(Y)
\end{equation*}
Now, $\tilde{S}(Y)$ is the test statistic for the correctly specified null hypothesis that median is $\mu$. Let the corresponding test by $\tilde{\phi}$. By Theorem \ref{theorem--exactness}, $E[\tilde{\phi}] = \alpha$. Recall that the randomization distribution is $\text{Binomial}(q, \frac{1}{2})$, regardless of the exact value of the null hypothesis. As such, 
\begin{align*}
S(Y) \geq S_\mu(Y) & \quad  \Rightarrow \quad  \phi_{1,n}(Y) \geq \tilde{\phi}  \\
& \quad  \Rightarrow \quad E_\mu[\phi_n(Y(\mu_0))] \geq E_\mu[\phi_n(Y(\mu))] = \alpha
\end{align*}

\subsection{Proof for Proposition \ref{theorem--equicorrelation}}
Given proposition \ref{proposition--generalpower}, we consider $P_Y(r)$. 

When $Y$ is an equi-correlated normal random variable, $P_Y$ depends only on the parameters $q$ and $\rho$. We denote it $P_{q, \rho}(r)$ for convenience. First note if $q = 2$, we know by explicit formula that the power is
	\begin{equation*}
		2 \gamma \mathbb{P}\left(\text{sgn}(Y) = (1, 1)\right) = \frac{\gamma}{2}\left[ 1 + \frac{2}{\pi} \arcsin \rho \right] > \frac{\gamma}{2} = \alpha~.
	\end{equation*}
	It is also immediate that power is strictly increasing in $\rho$.

	Now, suppose $q \geq 3$. Under assumption \ref{assumption--equicorrelation}, we have from equation (2.16) in \cite{gb2009} that for $\rho \geq 0$, 
	\begin{align}
		\frac{1}{2}P_{q, \rho}(r) & = {q \choose r}\int_{\mathbf{R}} \phi(y) \left( 1 - \Phi\left( -\frac{\sqrt{\rho}y}{\sqrt{1-\rho}} \right) \right)^r \Phi\left(-\frac{\sqrt{\rho}y}{\sqrt{1-\rho}} \right)^{q-r} dy \label{equation--gb2009} \\
			& = {q \choose r}\int_{\mathbf{R}} \phi(z) \left( 1 - \Phi\left(\tilde{\rho}z \right) \right)^r \Phi\left( \tilde{\rho} z \right)^{q-r} dz \mbox{ by symmetry of $\phi$ and setting } z = -y \nonumber \\
			& = {q \choose r}\int_{\mathbf{R}} \phi(z) \sum_{a=0}^r {r \choose a} (-1)^a \Phi\left(\tilde{\rho}z\right)^{q-r-a} dz \nonumber 
	\end{align}
	where $\tilde{\rho} = \frac{\sqrt{\rho}}{\sqrt{1-\rho}}$, with $\frac{d\tilde{\rho}}{d\rho} > 0$. Note again that as $\rho \to 1$, from equation \ref{equation--gb2009} we see that $P_{q, \rho}(r) \to 0$ if $r \neq 0$. On the other hand, $P_{q, \rho}(0) \to 1$
	
	To apply proposition (\ref{proposition--generalpower}), we first let $\gamma = 1$. Then,
	\begin{align*}
		& \mathbb{E}[\phi_n] = \sum_{r=0}^m P_{q, \rho}(r)  = 2 \int_{\mathbf{R}} \phi(z) \sum_{r=0}^m {q \choose r} \sum_{a=0}^r {r \choose a} (-1)^a \Phi\left(\tilde{\rho} z \right)^{q-r+a} dz
	\end{align*}
	The inner sum is:
	\begin{align*}
		& {q \choose 0}\left( {0 \choose 0} \Phi(\tilde{\rho}z)^q  \right) \\
		+ & {q \choose 1}\left( {1 \choose 0} \Phi(\tilde{\rho}z)^{q-1} - {1 \choose 1} \Phi(\tilde{\rho}z)^q  \right) \\
		+ & {q \choose 2}\left( {2 \choose 0} \Phi(\tilde{\rho}z)^{q-2} - {2 \choose 1} \Phi(\tilde{\rho}z)^{q-1} + {2 \choose 2} \Phi(\tilde{\rho}z)^q  \right) \\
		+ & ... \\
		+ & {q \choose m}\left( {m \choose 0} \Phi(\tilde{\rho}z)^{q-m} - {m \choose 1} \Phi(\tilde{\rho}z)^{q-m+1} + ... + {m \choose m} (-1)^{q} \Phi(\tilde{\rho}z)^q  \right)
	\end{align*}
	Observe that the coefficient of $\Phi(\tilde{\rho}z))^{q-r}$ is:
	\begin{equation*}
		\sum_{k = 0}^{m-r} {q \choose {r+k}}{ {r+k} \choose k } (-1)^k = \sum_{k = r}^{m} {q \choose k}{ k \choose k-r } (-1)^{k-r}
	\end{equation*}
	where the equality follows from replacing the index $k$ with $k+r$. Then, 
	\begin{align*}
		\sum_{k = r}^{m} {q \choose k}{ k \choose k-r } (-1)^{k-r} & = \frac{(-1)^{m-r} (m-r+1) }{q-r} {q \choose {m+1}} {{m+1} \choose {m-r+1}} \\
		& = \frac{(-1)^{m-r}}{q-r} {q \choose {m+1}} (m+1) {m \choose r}
	\end{align*}
	where the second step follows because $k {n \choose k} = n {{n-1 \choose {k-1}}}$ and ${m \choose {m-r}} = {m \choose r}$. The first step obtains by induction as follows. The base case with $m = r$ is immediate. For the induction step, we have that:
	\begin{align*}
		& \sum_{k = r}^{m+1} {q \choose k}{ k \choose k-r } (-1)^{k-r} \\
		& = \frac{(-1)^{m-r} (m-r+1) }{q-r} {q \choose {m+1}} {{m+1} \choose {m-r+1}} + {q \choose {m+1}}{ {m+1} \choose {m+1-r} } (-1)^{m+1-r} \\
		& = (-1)^{m-r} {q \choose {m+1}}{ {m+1} \choose {m+1-r} }\frac{m+1-q}{q-r} \\
		& = \frac{(-1)^{m+1-r}}{q-r}\frac{q-m-1}{m+1}{q \choose {m+1}} \frac{m+2-r}{m-r} { {m+1} \choose {m+1-r} } \\
		& = \frac{(-1)^{m+1-r} (m-r+2) }{q-r} {q \choose {m+2}} {{m+2} \choose {m-r+2}}~.
	\end{align*}
	This gives us the desired result. Substituting this into our expression for power, 
	\begin{align*}
		\sum_{r=0}^m P_{q, \rho}(r) &  = 2\int_{\mathbf{R}} \phi(z) \sum_{r=0}^m {q \choose r} \sum_{a=0}^r {r \choose a} (-1)^a \Phi\left(\tilde{\rho} z \right)^{q-r+a} dz \\
		& = 2\int_{\mathbf{R}} \phi(z) {q \choose {m+1}} (m+1) \sum_{r=0}^m   \frac{(-1)^{m-r}}{q-r} {m \choose r} \Phi(\tilde{\rho} z)^{q-r} dz \\
		& = 2\int_{0}^\infty \phi(z) {q \choose {m+1}} (m+1) \sum_{r=0}^m   \frac{(-1)^{m-r}}{q-r} {m \choose r} \Phi(\tilde{\rho} z)^{q-r} dz \\
		& \qquad + 2\int_{0}^\infty \phi(z) {q \choose {m+1}} (m+1) \sum_{r=0}^m   \frac{(-1)^{m-r}}{q-r} {m \choose r} \left(1-\Phi(\tilde{\rho} z)\right)^{q-r} dz  \\
		& = 2\int_{0}^\infty \phi(z) {q \choose {m+1}} (m+1) \sum_{r=0}^m   \frac{(-1)^{m-r}}{q-r} {m \choose r} \left(\Phi(\tilde{\rho} z)^{q-r} + \left(1-\Phi(\tilde{\rho} z)\right)^{q-r}  \right)dz
	\end{align*}
	where the third equality follows by a change of variable and the fact that $\Phi(-\tilde{\rho}z) = 1-\Phi(\tilde{\rho}z)$.
	Define the function:
	\begin{equation*}
		f(\theta) := \sum_{r=0}^m   \frac{(-1)^{m-r}}{q-r} {m \choose r} \theta^{q-r}
	\end{equation*}
	Suppose $q \geq 3$, we compute:
	\begin{align*}
		f'(\theta) & = \sum_{r=0}^m \frac{(-1)^{m-r}}{q-r} {m \choose r} (q-r)\theta^{q-r-1} \\
		& = \theta^{q-m-1} \sum_{r=0}^m (-1)^{m-r} {m \choose r} \theta^{m-r} = \theta^{q-m-1} (1-\theta)^m
	\end{align*}
%	\begin{align*}
%		f^{\prime \prime}(\theta) & = (q-m-1)\theta^{q-m-2}(1-\theta)^m + \theta^{q-m-1}(1-\theta)^{m-1}m \\
%		& = \theta^{q-m-2}(1-\theta)^{m-1}[(q-m-1) - \theta(q-2m-1)  ] \geq 0 \mbox{ for all } \theta \in [0,1]
%	\end{align*}
	Using this expression, we can write:
	\begin{align*}
		\frac{\partial }{\partial \rho}\sum_{r=0}^m P_{q, \rho}(r) &  = 2\int_{0}^\infty \phi(z) {q \choose {m+1}} (m+1) \frac{\partial }{\partial \rho}\left(f(\Phi(\tilde{\rho}z))  + f(1-\Phi(\tilde{\rho}z))  \right)  dz \\
		& =  2\int_{0}^\infty z\phi(z)\phi(\tilde{\rho}z) \frac{d \tilde{\rho}}{d \rho}{q \choose {m+1}} (m+1)  \Phi(\tilde{\rho} z)^{m}\left(1-\Phi(\tilde{\rho} z)\right)^{m}\\
		& \qquad \qquad \qquad \qquad \qquad \qquad \cdot \left( \Phi(\tilde{\rho}z)^{q-2m-1} - \left(1-\Phi(\tilde{\rho}z)\right)^{q-2m-1}   \right)  dz \\
		& \geq 0  \mbox{ if } q - 2m - 1 > 0
	\end{align*}
	The inequality is strict if $\rho > 0$ and $q-2m-1>0$. This restriction requires that $\alpha < 1 - \frac{1}{2^{q-1}}{q \choose \lceil q/2 \rceil}$ It is easy to check that  $\lfloor 1 - \frac{1}{2^{q-1}}{q \choose \lceil q/2 \rceil} \rfloor$ attains its minimum value at $3/4$ when $q = 4$. Hence, the bound applies if $\alpha < 1/4$. This gives us that power is strictly increasing in $\rho$ for all $\alpha \leq 1/4$ and $q \geq 3$. Finally, to show the size distortion, note that
	\begin{align*}
		E_\rho[\phi_n] & = E_{0}[\phi_n] + \int_{0}^{\rho}  \frac{\partial }{\partial \rho}\sum_{r=0}^m P_{q, \rho}(r)    \,    du \\
		& > \alpha \quad \mbox{ for } \rho > 0.5
	\end{align*} 
	This is because $E_{0}[\phi_n] = \alpha$ by construction and the second term is an integral of a function that is almost everywhere strictly positive. 
	
	To complete the proof, note that when $\gamma \neq 1$, we are linearly interpolating $\sum_{r=0}^m P_{q, \rho}(r) $ and $\sum_{r=0}^{m-1} P_{q, \rho}(r)$. As such, our result continues to hold. 

\subsection*{Proof of Proposition \ref{theorem--minimalcorrelation}}

	To apply Proposition \ref{proposition--generalpower}, we compute $P_{q, \rho}(r)$ under assumption \ref{assumption--minimalcorrelation}. With $q = 2$, minimal correlation is the same as equi-correlation. Suppose $q \geq 3$. Then for $r \geq 2$, 
	\begin{align*}
		\frac{1}{2}P_{q, \rho}(r)  = & \left( \frac{1}{2} \right)^{r} \cdot \left( \frac{1}{2} \right)^{q-2-r} \cdot \frac{1}{4}\left(1 + \frac{2}{\pi} \arcsin \rho   \right) {q-2 \choose r} \\
		& + \left( \frac{1}{2} \right)^{r-1} \cdot \left( \frac{1}{2} \right)^{q-2-(r-1)} \cdot \frac{1}{4}\left(1 - \frac{2}{\pi} \arcsin \rho   \right) {q-2 \choose r-1} \cdot {2 \choose 1} \\
		& + \left( \frac{1}{2} \right)^{r-2} \cdot \left( \frac{1}{2} \right)^{q-2-(r-2)} \cdot \frac{1}{4}\left(1 + \frac{2}{\pi} \arcsin \rho   \right) {q-2 \choose r-2} \\
		= & \left( \frac{1}{2} \right)  ^q \bigg[ {q-2 \choose r} + {q-2 \choose r-1} \cdot {2 \choose 1}  + {q-2 \choose r-2}  \\
		& + \left( {q-2 \choose r} - {q-2 \choose r-1} \cdot {2 \choose 1}  + {q-2 \choose r-2} \right) \frac{2}{\pi} \arcsin \rho   \bigg] \\
		= & \left( \frac{1}{2} \right)  ^q \bigg[ {q \choose r} + \left( {q-2 \choose r} - {q-2 \choose r-1} \cdot {2 \choose 1}  + {q-2 \choose r-2} \right) \frac{2}{\pi} \arcsin \rho   \bigg]
	\end{align*}
	where we have simply enumerated the cases in which 0, 1 or 2 of the correlated sub-clusters land on heads respectively. Similarly, 
	\begin{align*}
		\frac{1}{2}P_{q, \rho}(1) & = \left( \frac{1}{2} \right)  ^q \bigg[ {q \choose 1} + \left( {q-2 \choose 1} - {q-2 \choose 0}\cdot {2 \choose 1} \right)\frac{2}{\pi} \arcsin \rho   \bigg] \\
		\frac{1}{2}P_{q, \rho}(0) & = \left( \frac{1}{2} \right)  ^q \bigg[ {q \choose 0} + {q-2 \choose 0}\frac{2}{\pi} \arcsin \rho   \bigg]
	\end{align*}
	By proposition \ref{proposition--generalpower}, we write:
	\begin{align*}
		& \mathbb{E}[\phi_n]  = \sum_{r = 0}^{m-1} P_{q, \rho}(r) + \gamma P_{q, \rho}(m) \\
		& = \left(\frac{1}{2}\right)^{q-1} \left[  \gamma {q \choose m} + \sum_{r = 0}^{m-1} {q \choose r}    \right]  \\ \tiny
		& + \left(\frac{1}{2}\right)^{q-1} \left(  \frac{2 \arcsin \rho}{\pi}  \right) \left( \gamma \left( {q \choose m}  - 4 {q-2 \choose m-1} \right) +      1 + {q-2 \choose 1} - 2 + \sum_{r = 3}^{m-1}   \left( {q \choose r}  - 4 {q-2 \choose r-1} \right)   \right)
	\end{align*}
	The first term above is $\alpha$ by construction. It remains to show that the second term is positive. To do, we first evaluate it when $\gamma = 1$. The sum is telescoping:
	\begin{align*}
	W(m) :& = \, 1 + {q-2 \choose 1} - 2 + \sum_{r = 3}^{m}   \left( {q \choose r}  - 4 {q-2 \choose r-1} \right) \\
		& = {q-2 \choose 0}   \\
		& + {q-2 \choose 1} - {q-2 \choose 0}  - {q-2 \choose 0} 	 \\
		& + {q-2 \choose 2} - {q-2 \choose 1}  - {q-2 \choose 1} 	  +{q-2 \choose 0}\\
	 	& + {q-2 \choose 3} - {q-2 \choose 2}  - {q-2 \choose 2} 	  +{q-2 \choose 1}\\
	 	& + \, ... \\
	 	& +{q-2 \choose m} - {q-2 \choose m-1}  - {q-2 \choose m-1} 	  +{q-2 \choose m-2}
	\end{align*}
Cancelling terms diagonally, we have that:
\begin{align*}
	W(m) & = {q-2 \choose m} - {q-2 \choose m-1} \\
	& = \frac{(q-2)!}{m! (q-m-2)!} - \frac{(q-2)!}{(m-1)!(q-m-1)!} \\
	& = \frac{(q-2)!}{(m-1)!(q-m-2)!}\left[\frac{1}{m} - \frac{1}{q-m-1}\right]
\end{align*}
which is positive if and only if $m \leq \frac{q-1}{2}$. Furthermore, the term is strictly positive if $m < \frac{q-1}{2}$. Since $\arcsin \rho$ is strictly increasing in $\rho$, we have that size is strictly increasing in $\rho$.  Our goal was to show that:
\begin{equation*}
(1-\gamma) W(m-1) + \gamma W(m) > 0
\end{equation*}
By the above argument, for $m \leq \frac{q-1}{2}$, we have that $W(m-1) \geq 0$, $W(m) \geq 0$. The desired inequality obtains. 

To show (b), observe that under our assumption,
\begin{equation*}
	P\left(S(Y) > \tilde{S}(Y) \, \big\lvert \, \text{sgn}(Y) \right) > 0
\end{equation*}
In particular, $S(Y) > \tilde{S}(Y)$ when $Y_i \in \{\mu_0, \mu\}$, which occurs with strictly positive probability when $Y_i$ is continuously distributed. 

As before, (c) follows straightforwardly from the fact that $P(\text{sgn}(Y_i) = 1) \to 1$ as $\mu \to \infty$. 

\end{document}